\documentclass[11pt,a4paper]{article}

\usepackage[utf8]{inputenc}

\usepackage{fullpage}

\usepackage[bookmarks]{hyperref}
\hypersetup{colorlinks=true,citecolor=blue,linkcolor=blue,filecolor=blue,urlcolor=blue}

\usepackage{amsmath,amssymb,amsfonts,amsthm,amstext}
\usepackage{dsfont}
\usepackage{enumerate}

\numberwithin{equation}{section}

\usepackage{mathtools}

\usepackage[toc,page]{appendix}

\usepackage[usenames,dvipsnames]{xcolor}
\definecolor{darkgreen}{RGB}{0,153,0}
\usepackage{tikz}
\usepackage{pgfplots}
\usepackage[font=small]{caption}

\allowdisplaybreaks[2]

\usepackage{cleveref}

\newtheorem{theorem}{Theorem}[section]
\newtheorem{lemma}[theorem]{Lemma}
\newtheorem{proposition}[theorem]{Proposition}
\newtheorem{corollary}[theorem]{Corollary}
\theoremstyle{definition}
\newtheorem{definition}[theorem]{Definition}
\newtheorem{remark}[theorem]{Remark}

\newcommand\numberthis{\addtocounter{equation}{1}\tag{\theequation}}

\newcommand{\cH}{\mathcal{H}}
\newcommand{\cD}{\mathcal{D}}
\newcommand{\cS}{\mathcal{S}}
\newcommand{\cP}{\mathcal{P}}
\newcommand{\cT}{\mathcal{T}}
\newcommand{\cL}{\mathcal{L}}
\newcommand{\cB}{\mathcal{B}}
\newcommand{\cO}{\mathcal{O}}
\newcommand{\cU}{\mathcal{U}}
\newcommand{\cC}{\mathcal{C}}
\newcommand{\cE}{\mathcal{E}}
\newcommand{\cX}{\mathcal{X}}

\newcommand{\kE}{\mathfrak{E}}
\newcommand{\kD}{\mathfrak{D}}
\newcommand{\cV}{\mathcal{V}}

\newcommand{\bF}{\bar{F}}

\newcommand{\one}{\mathds{1}}
\newcommand{\eps}{\varepsilon}
\newcommand{\X}{\rangle\langle}

\DeclareMathOperator{\tr}{Tr}

\DeclareMathOperator{\id}{id}

\newcommand{\sumi}{\sum\nolimits}
\DeclareMathOperator{\supp}{supp}

\newcommand{\invP}[1]{\Phi^{-1}\left(#1\right)}

\newcommand{\nin}{n\in\mathbb{N}}

\newcommand{\nlim}{{n\rightarrow\infty}}
\newcommand{\kEm}{\kE_{\text{\normalfont mix}}}

\newcommand{\ul}[1]{\underline{#1}}
\newcommand{\ox}{\otimes}
\DeclareMathOperator{\spn}{span}
\newcommand{\dl}{d\mu(\lambda)}

\usepackage[backend=bibtex,firstinits=true,doi=false,isbn=false,url=false,maxbibnames=5]{biblatex}
\renewbibmacro{in:}{}
\begin{document}
%\title{\textbf{Source coding for a mixed source: determination of second order asymptotics}}
\title{\textbf{Second order asymptotics of \\visible mixed quantum source coding\\ via universal codes}}
\author{Felix Leditzky and Nilanjana Datta\\[0.25cm]
\textit{\small Statistical Laboratory, Centre for Mathematical Sciences, University of Cambridge,}\\
\textit{\small Wilberforce Road, Cambridge CB3 0WA, United Kingdom} }
\maketitle
\begin{abstract}
The simplest example of a quantum information source with memory is a mixed source which emits signals entirely from one of two memoryless quantum sources with given a priori probabilities. 
Considering a mixed source consisting of a general one-parameter family of memoryless sources, we derive the second order asymptotic rate for fixed-length visible source coding.
Furthermore, we specialize our main result to a mixed source consisting of two memoryless sources.
Our results provide the first example of second order asymptotics for a quantum information-processing task employing a resource with memory.
For the case of a classical mixed source (using a finite alphabet), our results reduce to those obtained by Nomura and Han \cite{NH13}.
To prove the achievability part of our main result, we introduce universal quantum source codes achieving second order asymptotic rates. 
These are obtained by an extension of Hayashi's construction \cite{Hay08} of their classical counterparts.
\end{abstract}

\section{Introduction}
Source coding (or data compression) is essential for efficient storage and transmission of information. 
Hence, evaluating the optimal rate of data compression is a fundamental problem in information theory. 
In classical information theory, the simplest class of sources is composed of so-called i.i.d.~or stationary, memoryless sources, the name `memoryless' arising from the fact that there is no correlation between successive signals emitted by such a source. 
Although these sources play a prominent role in information theory, in real-world applications the assumption of sources being memoryless is not necessarily justified. 
This is why it is important to study data compression for sources with memory. 
The simplest example of such a source is a {\em{mixed source}}.
It can be constructed from two i.i.d.~sources as follows. 
One associates a priori probabilities, say $t$ and $(1-t)$, to the two sources respectively. 
Then the mixed source is one for which all successive signals are emitted from the first source with probability $t$, or from
the second source with probability $(1- t)$. 
The memory of the mixed source can be trivially seen to be governed by a two-state Markov chain which is aperiodic but not irreducible, and hence such a source is non-ergodic (see e.g.~\cite{Nor98}).

Optimal rates of reliable data compression for the above sources and their quantum analogues were originally evaluated under the requirement that the error incurred in the compression and decompression scheme vanishes in the asymptotic limit (i.e.~the limit $n\rightarrow \infty$ where $n$ denotes the number of uses of the source). 
The optimal asymptotic rate for a classical i.i.d.~source is given by its Shannon entropy \cite{Sha48}, whereas the corresponding rate for a quantum memoryless source is given by its von Neumann entropy \cite{Sch95}. 
The optimal (first order) asymptotic rate for mixed source coding was derived by Han \cite{Han03} in the classical case, and in \cite{BD06a} in the quantum case, employing the so-called Information Spectrum Approach.\footnote{This approach provides a unifying mathematical framework for obtaining asymptotic rate formulae for various different tasks in information theory, without making any assumptions on the structure or properties of the underlying resources.}
It was shown to be given by the maximum of the Shannon (resp.~von Neumann) entropies of the two underlying classical (resp.~quantum) memoryless sources.

Recently, a more refined asymptotic analysis of data compression for memoryless sources under the (more reasonable) requirement of a non-zero error threshold $\eps \in (0,1)$ was done (\cite{DL14b}, see also \cite{TH13}). 
The quantity analysed was the {\em{minimum compression length}}, which we denote by $\log M_n\equiv \log_2 M_n$. 
In the classical case this is the minimum number of bits needed to compress signals emitted by $n$ uses of the source so that they can be recovered with an error of at most $\eps$ upon decompression. 
In the quantum case, it is the minimum dimension of the compressed Hilbert space compatible with the given error threshold. 
The second order asymptotic expansions of the minimum compression length for both the classical and quantum cases were proved to be of the form
\begin{align}
\log M_n = a n + b \sqrt{n} + {\mathcal{O}}(\log n).\label{eq:mcl-asymptotic expansion}
\end{align}
Here, the coefficient $a$ of the leading order term constitutes the first order asymptotics of the minimum compression length, and, as expected, is given by the optimal asymptotic rate. 
The coefficient $b$ is a function of both the source and the allowed error threshold $\eps$. 
It constitutes the second order asymptotics and is hence referred to as the {\em{second order asymptotic rate}} (cf.~\Cref{def:source-coding}). It is given by $-\sqrt{V}\Phi^{-1}(\eps)$, where $\Phi^{-1}$ denotes the inverse of the cumulative distribution function of the standard normal distribution (defined in \eqref{eq:invP}), and $V$ denotes the {\em{information variance}} of the source (cf.~\Cref{def:quantum-relative-entropy}(ii)). 
The asymptotic expansion \eqref{eq:mcl-asymptotic expansion} was evaluated for fixed-length source coding in the classical case by Strassen \cite{Str62} (see also Hayashi \cite{Hay08}) and in the quantum case (for the {\em{visible setting}}) in \cite{DL14b}. 

Deriving second order asymptotic rates in Classical Information Theory was initiated by Strassen \cite{Str62}. In Quantum Information Theory, the topic was introduced in 2012 independently by Li \cite{Li14} and Tomamichel and Hayashi \cite{TH13}, who obtained a second order asymptotic characterization of hypothesis testing.
In the latter paper, the authors used this result to characterize the second order asymptotics of randomness extraction and source compression with quantum side information. 
Since then second order asymptotic expansions have been obtained for a range of operational quantities characterizing information-processing tasks. 
These include entanglement conversion \cite{KH13a,DL14b}, classical-quantum channel coding \cite{TT13,BG13,DL14b}, quantum source coding \cite{DL14b}, source coding with quantum side information \cite{TH13,BG13}, noisy dense-coding \cite{DL14b}, achievability bounds on the coding rate for entanglement-assisted communication \cite{DTW14}, an achievability bound on the quantum communication cost in state redistribution \cite{DHO14}, and achievability bounds on the quantum capacity \cite{BDL15,TBR15}. 
Common to all these endeavours is that the underlying resource (such as the source state in source coding, or the channel in classical-quantum channel coding) is assumed to be memoryless.

Obtaining second order asymptotic expansions for any information-processing task employing resources with memory is a more challenging task.
The first foray into this task was made in classical information theory by Polyanskiy, Poor and Verdú \cite{PPV11}, who obtained second order expansions for the capacity of a classical mixed channel (see also \cite{TT14}).
In \cite{NH13}, Nomura and Han evaluated second order optimal rates for fixed-length source coding for a classical mixed source (see also \cite{Hay08}).
Yagi and Nomura \cite{YN14} (see also Yagi, Han, Nomura \cite{YHN15}) derived the second order coding rate, or channel dispersion, of a mixed channel under the assumption that the channel is \emph{well-ordered} (cf.~\cite[Def.~3]{YN14} or \cite[Def.~3]{YHN15}).

All the works mentioned above emphasize the importance of mixed source coding or mixed channel coding as simple yet instructive examples of an information-theoretic task employing non-ergodic resources.
The main focus of our paper is to extend the analysis of such tasks to the quantum regime, by investigating mixed quantum source coding.
We consider fixed-length source coding for a mixed source constructed from a general one-parameter family of memoryless sources, obtaining optimal second order rates in the visible setting.
In the classical case, our results reproduce the optimal rates of Nomura and Han in the finite-alphabet setting. 
The key tool in our derivations is the second order asymptotic expansion of the information spectrum entropy $D_s^\eps(\rho\|\tau)$ (see \eqref{eq:inf-spec-rel-entropy} for a definition), which was derived in \cite{TH13}.
To prove achievability of the second order asymptotic rates, we introduce universal quantum source codes achieving second order asymptotic rates. These universal codes are obtained by extending the original construction of universal quantum source codes by Jozsa \textit{et al.~}\cite{JHHH98} using Hayashi's construction of classical universal source codes which achieve second order asymptotic rates \cite{Hay08}.

The paper is organized as follows. 
After setting the notation and providing the necessary mathematical prerequisites in \Cref{sec:preliminaries}, we discuss the operational setting of mixed source coding in \Cref{sec:operational-setting}:
In \Cref{sec:mixed-sources} we explain in detail how a mixed source consisting of a one-parameter family of memoryless sources is constructed.
\Cref{sec:visible-source-coding} gives a short overview of visible quantum source coding. 
In \Cref{sec:second-order-rates} we define the second order asymptotic rate of a quantum source. 
Our main result is given in Section~\ref{sec:results} and comprises expressions for the second order asymptotic rates of mixed source coding. 
The proofs of these expressions are given in Section~\ref{sec:proofs}. For the achievability proofs, we construct universal source codes achieving second order rates in \Cref{sec:universal}. Finally, in \Cref{sec:conclusion} we present a conclusion and mention open problems.

\section{Mathematical preliminaries}\label{sec:preliminaries}
For a Hilbert space $\cH$, let $\cB(\cH)$ denote the algebra of linear operators acting on $\cH$, and let  $\cP(\cH)$ denote
the set of positive semi-definite operators on $\cH$. 
Further, let $\cD(\cH)\coloneqq\lbrace\rho\in\cP(\cH)\colon \tr\rho=1\rbrace$ denote the set of states (density matrices) on $\cH$. 
For a state $\rho\in\cD(\cH)$, the von Neumann entropy $S(\rho)$ is defined as $S(\rho)\coloneqq - \tr\left(\rho\log\rho\right)$. 
Here and henceforth, all logarithms are taken to base $2$, and all Hilbert spaces are assumed to be finite-dimensional. 
We denote by $\one\in\cP(\cH)$ the identity operator on $\cH$, and by $\id\colon \cB(\cH)\rightarrow\cB(\cH)$ the identity map on operators on $\cH$. 
For a pure state $|\psi\rangle$, the corresponding projector is abbreviated as $\psi\equiv|\psi\X\psi|$. 

A quantum operation $\Lambda\colon \cD(\cH)\rightarrow\cD(\cH')$ is a linear, completely positive, trace-preserving (CPTP) map. 
For self-adjoint operators $A, B \in \cB(\cH)$, let $\lbrace A \geq B\rbrace$ denote the projector onto the subspace spanned by the eigenvectors of the operator $A-B$ corresponding to non-negative eigenvalues, and set $\lbrace A < B\rbrace \coloneqq \one - \lbrace A \geq B \rbrace.$ 
We further define $A_+ \coloneqq \lbrace A\geq 0\rbrace A\lbrace A\geq 0\rbrace$ and take note of the following property:
\begin{lemma}[\cite{ON00}]\label{lem:tr-projector}
For operators $A,B\geq 0$ and $0\leq P \leq\one$ we have
	\begin{align*}
\tr(A - B)_+ =	\tr[\lbrace A\geq B\rbrace(A-B)]\geq\tr [P(A-B)].
	\end{align*}
\end{lemma}

The inverse of the cumulative distribution function (c.d.f.) of a standard normal random variable is defined by
\begin{align}
\Phi^{-1}(\eps) \coloneqq \sup\lbrace z\in\mathbb{R}\colon\Phi(z)\leq \eps\rbrace,\label{eq:invP}
\end{align}
where $\Phi(z) = \frac{1}{\sqrt{2\pi}}\int_{-\infty}^z e^{-t^2/2}dt$. Note that $\Phi(x) = 1-\Phi(-x)$ and $\invP{1-x} = -\invP{x}$. 
%The fact that $\Phi^{-1}$ is continuously differentiable ensures that the following lemma holds:
%\begin{lemma}[\cite{DL14b}]\label{lem:phi-trick}
%Let $\eps>0$, then $$\sqrt{n}\,\Phi^{-1}\left(\eps\pm\frac{1}{\sqrt{n}}\right) = \sqrt{n}\,\Phi^{-1}(\eps)\pm \left(\Phi^{-1}\right)'(\xi)$$ for some $\xi$ with $|\xi-\eps|\leq\frac{1}{\sqrt{n}}$.
%\end{lemma}

Two central quantities in our discussion are the quantum relative entropy $D(\rho\|\tau)$ and the quantum information variance $V(\rho\|\tau)$:
\begin{definition}\label{def:quantum-relative-entropy}
			Let $\rho\in\cD(\cH)$ and $\tau\in\cP(\cH)$.
			\begin{enumerate}[{\normalfont (i)}]
			\item \cite{Ume62} The \emph{quantum relative entropy} $D(\rho\|\tau)$ is defined as 
						\begin{align*}
						D(\rho\|\tau) \coloneqq  \begin{cases}\tr[\rho(\log\rho-\log\tau)] & \text{if }\supp\rho\subseteq\supp\tau\\ \infty & \text{else.}\end{cases}
						\end{align*}
						Note that the von Neumann entropy is given by $S(\rho)=-D(\rho\|\one)$.
			\item \cite{TH13} The \emph{quantum information variance} $V(\rho\|\tau)$ is defined as
						\begin{align*}
						V(\rho\|\tau) \coloneqq  \tr\left[\rho(\log\rho-\log\tau)^2\right]-D(\rho\|\tau)^2.
						\end{align*}
Further, we define
						$\sigma(\rho\|\tau)\coloneqq \sqrt{V(\rho\|\tau)}$ and 
						\begin{align}
						\sigma(\rho)\coloneqq \sigma(\rho\|\one) = \sqrt{V(\rho\|\one)}\label{eq:sigma}.
						\end{align}
			\end{enumerate}
Note that $\sigma(\rho)$ is equal to the standard deviation of the probability distribution formed by the eigenvalues of $\rho$.
In the classical literature, the information variance of a source is sometimes also referred to as \emph{varentropy}.
\end{definition}

In \cite{TH13} the authors introduced the \emph{information spectrum relative entropy} $D_s^\eps(\rho\|\tau)$, defined for $\eps\in (0,1)$, $\rho\in\cD(\cH)$, and $\tau\in\cP(\cH)$ as
\begin{align}\label{eq:inf-spec-rel-entropy}
D_s^\eps(\rho\|\tau)\coloneqq \sup\lbrace \gamma\in\mathbb{R}\colon \tr\left(\rho\left\lbrace\rho\leq 2^\gamma \tau \right\rbrace\right) \leq \eps\rbrace.
\end{align}
This quantity is particularly useful because its second order asymptotic expansion can be employed to obtain the second order asymptotics of quantum hypothesis testing, as shown in \cite{TH13}.
The derivation of our main results is based on the second order asymptotic expansion of the information spectrum relative entropy, which we employ in the following form:
\begin{theorem}[\cite{TH13}]\label{thm:SOA}
Let $\rho\in\cD(\cH)$ with $S=S(\rho)$ and $\sigma=\sigma(\rho)$.
There is a $K>0$ such that for any $L\in\mathbb{R}$ and $n\in\mathbb{N}$ we have
\begin{align}\label{eq:thm-SOA}
\left| \tr\left(\rho^{\otimes n} \left\lbrace \rho^{\otimes n}\leq 2^{-nS+\sqrt{n} L}\one \right\rbrace\right) - \Phi\left(\frac{L}{\sigma}\right) \right| \leq \frac{K}{\sqrt{n}}.
\end{align}
\end{theorem}
Note however, that the trace expression on the left-hand side of \eqref{eq:thm-SOA} only depends on the eigenvalues of $\rho^{\otimes n}$. Hence, \Cref{thm:SOA} already follows from the second order asymptotics of classical source coding derived by Strassen \cite{Str62}.

\section{Operational setting}\label{sec:operational-setting}
\subsection{Mixed quantum sources}\label{sec:mixed-sources}
A general quantum information source is characterized by an ensemble $\kE=\lbrace p_i,|\psi_i\rangle\rbrace_i$ of pure states (or signals) $|\psi_i\rangle\in\cH$ which are emitted by the source with corresponding probabilities $p_i$. 
We refer to $\kE$ as the {\em{source ensemble}}, and the associated density matrix (or ensemble average state) $\rho=\sum_i p_i\psi_i$ is called the {\em{source state}}. A source is called \emph{memoryless} if there are no correlations between successive signals emitted by the source. Consequently, we can characterize $n$ uses of a memoryless source $\kE$ by the source ensemble $\kE^n=\left\lbrace p_{\underline{i}},|\psi_{\underline{i}}\rangle\right\rbrace_{\underline{i}}$ where ${\underline{i}}\coloneqq i_1i_2\ldots i_n$ is a sequence of indices of length $n$, and we define
\begin{align}\label{eq:index-notation}
p_{\underline{i}} \coloneqq p_{i_1}p_{i_2}\ldots p_{i_n} \qquad\text{and}\qquad |\psi_{\underline{i}}\rangle \coloneqq |\psi_{i_1}\rangle \otimes  |\psi_{i_2}\rangle \otimes \ldots  |\psi_{i_n}\rangle.
\end{align}
The corresponding source state for $n$ uses of the source $\kE$ is given by $\rho^{\otimes n}$. 

We now construct a mixed source consisting of memoryless sources.
To this end, let $\Lambda$ be an arbitrary parameter space with a normalized measure $\mu$, i.e.~$\int_\Lambda d\mu(\lambda) = 1$.
Consider a family of memoryless sources parametrized by $\lambda\in\Lambda$, with source ensemble $\kE_\lambda=\lbrace q_i^{(\lambda)},|\varphi^{(\lambda)}_i\rangle\rbrace_i$ and source state $\rho_\lambda=\sum_i q_i^{(\lambda)}\varphi^{(\lambda)}_i$.
The mixed source is the one that emits all successive signals from the memoryless source $\kE_\lambda$ according to the probability measure $d\mu(\lambda)$.
We denote the mixed source obtained from this construction by $(\rho_\lambda,\dl)_{\lambda\in\Lambda}$.
The source state $\rho^{(n)}$ for $n$ uses of $(\rho_\lambda,\dl)_{\lambda\in\Lambda}$ is given by
\begin{align}\label{eq:general-mixture}
\rho^{(n)} = \int_{\Lambda} \rho_{\lambda}^{\ox n} \dl,
\end{align}
and the corresponding (not necessarily finite) ensemble is given by 
\begin{align}\label{eq:general-mixture-ensemble}
\kEm^{(n)} \coloneqq \left\lbrace d\mu(\lambda) q_{\underline{i_\lambda}}^{(\lambda)}; |\varphi_{\underline{i_\lambda}}^{(\lambda)}\rangle\right\rbrace_{\underline{i_\lambda},\,\lambda\in\Lambda},
\end{align}
where $\underline{i_\lambda}$ is a sequence of indices of length $n$ and $|\varphi_{\underline{i_\lambda}}^{(\lambda)}\rangle$ is a tensor product of $n$ pure states as in \eqref{eq:index-notation} for each $\lambda\in\Lambda$.

Let us consider the special case where the measure $\mu$ has finite support on points $\lambda_1,\dots,\lambda_k\in\Lambda$, corresponding to a discrete probability distribution $\lbrace t_j\rbrace_{j=1}^k$.
Hence, we have $k$ memoryless quantum information sources with source ensembles $\kE_j=\lbrace q_i^{(j)},|\varphi^{(j)}_i\rangle\rbrace$ and source states $\rho_j=\sum_i q_i^{(j)}\varphi^{(j)}_i$ for $j=1,\dots, k$.
The underlying source ensemble for $n$ uses of this mixed source is
\begin{align}
\kEm^{(n)} \coloneqq \left\lbrace t_1q_{\underline{i_1}}^{(1)}, \dots,  t_kq_{\underline{i_k}}^{(k)}; |\varphi_{\underline{i_1}}^{(1)}\rangle,\dots, |\varphi_{\underline{i_k}}^{(k)}\rangle\right\rbrace_{\underline{i_j},\,j=1,\dots,k},\label{eq:mixed-source-ensemble}
\end{align}
and the source state is given by 
\begin{align}\label{eq:blocklength-n-source-state}
\rho^{(n)} \coloneqq \sum_{j=1}^k t_j \rho_j^{\otimes n}.
\end{align}
We denote such a discrete mixed source consisting of $k$ memoryless sources $\rho_1,\dots,\rho_k$ by the tuple $(\lbrace\rho_j\rbrace_{j=1}^k,\lbrace t_j\rbrace_{j=1}^k)$ or simply $(\rho_j,t_j)_{j=1}^k$. 
In the special case of two memoryless sources, $k=2$, we set $t\equiv t_1$ (such that $t_2=1-t$) and write $(\rho_1,\rho_2,t)$ for the resulting mixed source.
The source state for $n$ uses of the mixed source $(\rho_1,\rho_2,t)$ is given by $\rho^{(n)} = t\rho_1^{\ox n} + (1-t)\rho_2^{\ox n}$. 
The parameter $t$ is also referred to as \emph{mixing parameter}.

Finally, we also mention the special case of a mixed source where we have a fixed set of pure states $\lbrace|\varphi_i\rangle\rbrace_i$, and for $\lambda\in\Lambda$ the source $\kE_\lambda$ corresponds to a probability distribution $\lbrace q^{(\lambda)}_i\rbrace_i$ over the pure states $\lbrace|\varphi_i\rangle\rbrace_i$.
That is, in this case we have $\lbrace|\varphi^{(\lambda)}_i\rangle\rbrace_i = \lbrace|\varphi_i\rangle\rbrace_i$ for all $\lambda\in\Lambda$.
The source state $\rho_\lambda$ of the memoryless source $\kE_\lambda$ is then given by $\rho_\lambda=\sum_i q_i^{(\lambda)} |\varphi_i\rangle\langle \varphi_i|$.

\subsection{Quantum source coding}\label{sec:visible-source-coding}
In fixed-length quantum source coding the aim is to store the information emitted by the source in a {\em{compressed state}} $\rho_c \in \cD(\cH_c)$ with $\dim\cH_c < \dim\cH$, such that it can later be decompressed yielding a state which is sufficiently close to the source state $\rho$ with respect to some chosen distance measure. 

There are two different settings \cite{BCF+01,Hay02,Win99a} for the compression part of the protocol outlined above: {\em{visible}} and {\em{blind}}. 
In this paper we only consider the visible setting.\footnote{For a discussion of the blind setting and its comparison to the visible setting, see e.g.~\cite{BCF+01,Hay02,Win99a} or Section V.A in \cite{DL14b}.} 
In this setting, the compressor (say, Alice) knows the identity of the signals $\psi_i$. 
In fact, on each use of the source Alice receives classical information in the form of an index $i$ labelling the signal $\psi_i$ emitted by the source. 
She then uses an \emph{arbitrary} map $\cV\colon \lbrace i\rbrace\rightarrow \cD(\cH_c)$ to encode the signal $\psi_i$ in a state $\cV(i)\in\cD(\cH_c)$. 
We stress that $\cV$ (which we refer to as {\em{visible encoding}}) is not a CPTP map acting on the signals $\psi_i$; Alice simply prepares a quantum state $\cV(i)$ on receiving the index $i$. 
This is in contrast to the \emph{blind setting} of source coding, where the encoder does not have any knowledge about the pure states $\psi_i$ and is therefore required to apply a quantum operation $\cE$ to the source state $\rho$. 
Henceforth, we restrict the discussion to the visible setting. 
In the decompression part of the protocol, the compressed signal $\cV(i)$ is subjected to a quantum operation $\mathfrak{D}\colon \cD(\cH_c)\rightarrow\cD(\cH)$ which we call the \emph{decoding map}.

\subsection{Definition of the second order asymptotic rate}\label{sec:second-order-rates}
Our aim is to derive the \emph{second order asymptotic rate} (or in short, second order rate) for fixed-length visible quantum source coding of a mixed source, whose precise definition we give below. 
Since we only discuss the visible source coding setting in this paper, we will henceforth suppress the attribute `visible' in all definitions.
  
We choose the ensemble average fidelity as the figure of merit in our analysis of fixed-length quantum source coding, defined as follows:
\begin{definition}\label{def:figure-of-merit}
Let $\kE=\lbrace p_i,|\psi_i\rangle \rbrace_i$ be a pure-state ensemble with $|\psi_i\rangle \in \cH$ for all $i$. 
We say that the triple $\cC = (\cV,\kD,M)$ defines a \emph{code} for fixed-length visible source coding if $\cV\colon \lbrace i \rbrace \rightarrow \cD(\cH_c)$ is an arbitrary encoding map,  $\kD\colon \cD(\cH_c)\rightarrow \cD(\cH)$ is a decoding CPTP map, and $\cH_c$ is the compressed Hilbert space with $M\coloneqq \dim\cH_c < \cH$.

The \emph{ensemble average fidelity} $\bF(\kE,\cC)$ of the ensemble $\kE$ and the code $\cC$ is defined as
\begin{align*}
\bar{F}(\kE,\cC) \coloneqq \sumi_i p_i\tr((\kD\circ\cV)(i)\psi_i).
\end{align*}
For a mixed source $(\rho_\lambda,d\mu(\lambda))_{\lambda\in\Lambda}$ as defined in \Cref{sec:mixed-sources} with ensemble $\kEm$ given as in \eqref{eq:general-mixture-ensemble} for $n=1$, the ensemble average fidelity $\bF(\kEm,\cC)$ is correspondingly defined as
\begin{align*}
\bF(\kEm,\cC) \coloneqq \int_{\lambda\in\Lambda} d\mu(\lambda) \sumi_i q_i^{(\lambda)}\tr\left((\kD\circ\cV)(i)\psi_i^{(\lambda)}\right).
\end{align*}
\end{definition}

\noindent This leads to the following definition:
\begin{definition}\label{def:source-coding}
Let $(\rho_\lambda,d\mu(\lambda))_{\lambda\in\Lambda}$ be a mixed source, and let $\eps\in(0,1)$.
For $n\in\mathbb{N}$ let $\kEm^{(n)}$ as defined in \eqref{eq:general-mixture-ensemble} be the source ensemble for $n$ uses of the mixed source $(\rho_\lambda,d\mu(\lambda))_{\lambda\in\Lambda}$.
Given $R\in\mathbb{R}$, we say that any $r\in\mathbb{R}$ is an \emph{$(R,\eps)$-achievable rate} if there exists a sequence $\lbrace\cC_n\rbrace_{\nin}$ of codes $\cC_n=(\cV_n,\kD_n,M_n)$ such that
\begin{align}\label{eq:r-eps-achievable-rate}
\liminf_{n\rightarrow\infty} \bar{F}\left(\kEm^{(n)},\cC_n\right)\geq 1- \eps\qquad \text{and}\qquad \limsup_{n\rightarrow \infty}\frac{\log M_n - nR}{\sqrt{n}} \leq r.
\end{align}
The \emph{second order asymptotic rate} $b\left(R,\eps|\rho\right)$ for $n$ uses of the mixed source $(\rho_\lambda,d\mu(\lambda))_{\lambda\in\Lambda}$ is then defined as the infimum over all $(R,\eps)$-achievable rates $r$.
\end{definition}

\begin{remark}~\label{rem:second-order-implies-first-order}
\begin{enumerate}[(i)]
\item For any $R>0$ the quantity $b\left(R,\eps | \rho\right)$ is only finite if the parameter $R$ equals the optimal first order rate $a$ of the protocol, i.e.~a real number $a$ satisfying
\begin{align}\label{eq:first-order-expansion}
\log M_n = n a + f(n)
\end{align}
with $f(n)\in\cO(\sqrt{n})$. This can be seen as follows: Substituting \eqref{eq:first-order-expansion} in \eqref{eq:r-eps-achievable-rate} of \Cref{def:source-coding}(ii) yields
\begin{align}
\frac{n a - nR}{\sqrt{n}} + \frac{f(n)}{\sqrt{n}} = \sqrt{n}(a-R) + \frac{f(n)}{\sqrt{n}}.\label{eq:first-order-different}
\end{align}
Taking the limit superior in \eqref{eq:first-order-different}, the second term is some constant since $f(n)\in\cO(\sqrt{n})$, whereas the first term diverges to either $+\infty$ if $R < a$ or $-\infty$ if $R>a$.
\item For quantum source coding using a single memoryless source, $a$ is equal to the von Neumann entropy $S(\rho)$ of the source, and \eqref{eq:first-order-expansion} is proven in \cite{Sch95,Win99}.
\end{enumerate}
\end{remark}

\section{Main results}\label{sec:results}
Our main result is the derivation of the second order asymptotic rate for $n$ uses of a mixed source $(\rho_\lambda,d\mu(\lambda))_{\lambda\in\Lambda}$ with source state $\rho^{(n)} = \int_{\Lambda} \rho_\lambda^{\ox n} d\mu(\lambda)$ as defined in \Cref{sec:mixed-sources}.
In order to state our main result, we make the following definition: For a fixed $a>0$ let $\cL_=(a) \coloneqq \lbrace \lambda\in\Lambda\colon S(\rho_\lambda) = a\rbrace$ and $\cL_<(a) \coloneqq \lbrace \lambda\in\Lambda\colon S(\rho_\lambda) < a\rbrace$.
Furthermore, recall that for $\rho\in\cD(\cH)$ we set $\sigma(\rho)=\sqrt{V(\rho\|\one)}$ (cf.~\eqref{eq:sigma}).
We then have:
\begin{theorem}\label{thm:general-mixture}
Let $\Lambda$ be an arbitrary parameter space with a normalized measure $\mu$, that is, $\int_\Lambda d\mu(\lambda) = 1$, and let $(\rho_\lambda,\dl)_{\lambda\in\Lambda}$ be a mixed source. % with source state $\rho^{(n)}$ as defined in \eqref{eq:general-mixture}.
Furthermore, let $a>0$, $\eps\in(0,1)$, and define $\sigma_\lambda = \sigma(\rho_\lambda)$ for $\lambda\in\Lambda$.
Then the second order asymptotic rate $b(a,\eps|\rho)$ for $n$ uses of the mixed source $(\rho_\lambda,\dl)_{\lambda\in\Lambda}$ is the solution of the equation
\begin{align*}
\int_{\cL_=(a)}\Phi\left(\frac{b}{\sigma_{\lambda}}\right) \dl + \int_{\cL_<(a)}\dl = 1 - \eps.
\end{align*}
\end{theorem}

If the measure $\mu$ has finite support on points $\lambda_1,\dots,\lambda_k\in\Lambda$, \Cref{thm:general-mixture} reduces to the following 

\begin{corollary}\label{thm:identical-entropies}
Consider a mixed source $\rho=(\rho_j,t_j)_{j=1}^k$, and set $S_j=S(\rho_j)$ and $\sigma_j= \sigma(\rho_j)$ for $j=1,\dots,k$.
For $a>0$ and $\eps\in(0,1)$, the second order asymptotic rate $b(a,\eps|\rho)$ for $n$ uses of the mixed source $\rho=(\rho_j,t_j)_{j=1}^k$ is given by the solution of the equation
\begin{align*}
\sum_{i\colon S_i=a} t_i \Phi\left(\frac{L}{\sigma_i}\right) + \sum_{i\colon S_i<a} t_i = 1-\eps.
\end{align*}
\end{corollary} 

Finally, we consider the special case of a mixed source consisting of two memoryless sources $\rho_1,\rho_2\in\cD(\cH)$ with corresponding source state 
\begin{align*}
\rho^{(n)}=t\rho_1^{\otimes n} + (1-t)\rho_2^{\otimes n}
\end{align*} 
and mixing parameter $t\in(0,1)$. We adhere to the discussion of classical mixed source coding by Nomura and Han \cite{NH13} by considering the following three cases,\footnote{Note that the assumption $S_1>S_2$ in Cases 2 and 3 can be made without loss of generality. } abbreviating $S_i\equiv S(\rho_i)$ and $\sigma_i\equiv \sigma(\rho_i)$ for $i=1,2$:
\begin{enumerate}[\hspace{0.3cm}{Case} 1:]
\item $S_1=S_2$
\item $S_1>S_2, t > \eps$
\item $S_1>S_2, t < \eps$
\end{enumerate}
We state the second order rate in each of the three cases in the following theorem:

\begin{theorem}\label{thm:refined-analysis}
Consider a mixed source $\rho=(\rho_1,\rho_2,t)$ with $\rho_1,\rho_2\in\cD(\cH)$ and $t\in(0,1)$, and set $S_i\coloneqq S(\rho_i)$ and $\sigma_i\coloneqq \sigma(\rho_i)$ for $i=1,2$. For $\eps\in(0,1)$ the second order asymptotic rate $b(a,\eps|\rho)$ for $n$ uses of the mixed source $(\rho_1,\rho_2,t)$ is given by the following expressions:
\begin{enumerate}[{\normalfont (i)}]
\item For $S_1=S_2\equiv S$, we have $b(S,\eps|\rho) = L$ where $L$ is the solution of the equation
\begin{align}\label{eq:L-relation}
t\Phi\left(\frac{L}{\sigma_1}\right) + (1-t)\Phi\left(\frac{L}{\sigma_2}\right) = 1-\eps.
\end{align}

\item For $S_1>S_2$ and $t > \eps$, we have
\begin{align}
b(S_1,\eps|\rho) = -\sigma_1\invP{\frac{\eps}{t}}.
\end{align}
\item For $S_1>S_2$ and $t < \eps$, we have
\begin{align}
b(S_2,\eps|\rho) = -\sigma_2\invP{\frac{\eps-t}{1-t}}.
\end{align}
\end{enumerate}
\end{theorem}

\begin{remark}\label{rem:case-1-general}~
\begin{enumerate}[(i)]
\item Upon replacing the quantum sources $\rho_\lambda$ with classical i.i.d.~sources characterized by a random variable $Y_\lambda$, identifying $S(\rho_\lambda)$ with the Shannon entropy $H(Y_\lambda)$, and the quantum information variance $\sigma_\lambda$ with the standard deviation of the random variable $\log Y_\lambda$, \Cref{thm:general-mixture} and \Cref{thm:refined-analysis} reproduce Theorem 8.3 and Theorem 7.1 in \cite{NH13}, respectively, in the case of a finite source alphabet.

\item Recall from \Cref{rem:second-order-implies-first-order}(i) that the statement $b(S_1,\eps|\rho) = -\sigma_1\invP{{\eps}/{t}}<\infty$ in \Cref{thm:refined-analysis}(ii) implies that the first order rate equals $S_1$. In particular, in this case $b(S_2,\eps|\rho) = \infty$.
Similarly, in \Cref{thm:refined-analysis}(iii) the first order rate is given by $S_2$, and $b(S_1,\eps|\rho) = -\infty$.

\item To determine the range of $L$ in \Cref{thm:refined-analysis}(i), assume without loss of generality that $\sigma_1<\sigma_2$. Then, using properties of the c.d.f.~$\Phi$ of a normal distribution and definition \eqref{eq:L-relation} of $L$, it follows easily that
\begin{subequations}\label{eq:L-bounds}
\begin{align}
L &\in[-\sigma_1\invP{\eps},-\sigma_2\invP{\eps}] & &\text{if }\eps\in \left(0,1/2\right),\label{eq:L-bounds-left}\\
L&\in [-\sigma_2\invP{\eps},-\sigma_1\invP{\eps}] & &\text{if }\eps\in \left(1/2,1\right),\label{eq:L-bounds-right}
\end{align}
\end{subequations}
and $L=0$ for $\eps=1/2$. See \Cref{fig:L-plot} for a plot showing a typical example of this.
\end{enumerate}
\end{remark}

%We also prove the following generalized version of \Cref{thm:refined-analysis} for an arbitrary mixture of i.i.d.~sources, constructed as follows:
%Let $\Lambda$ be an arbitrary parameter space with a normalized measure $\mu$, i.e.~$\int_\Lambda d\mu(\lambda) = 1$, and for each $\lambda\in\Lambda$ let $\rho_\lambda\in\cD(\cH)$ be the source state of an i.i.d.~source.
%We denote the mixed source obtained from this construction by $(\rho_\lambda,\dl)_{\lambda\in\Lambda}$.
%The source state $\rho^{(n)}$ for $n$ uses of $(\rho_\lambda,\dl)_{\lambda\in\Lambda}$ is then given by
%\begin{align}\label{eq:general-mixture}
%\rho^{(n)} = \int_{\Lambda} \rho_{\lambda}^{\ox n} \dl.
%\end{align}

\begin{figure}[t]
\small
\centering
\begin{tikzpicture}[scale=1] 
    \begin{axis}[
    width = 0.6\textwidth,
    legend style = {at={(1,1)}, anchor = north east},
    legend cell align = left,
    axis lines = left,
    grid=major, % Display a grid
    grid style={gray!25}, % Set the style
    xlabel = $\eps$,
    xmin = 0,
    xmax = 1,
%    every axis y label/.style = {at={(ticklabel* cs:0.9)},anchor=east},
%    every axis x label/.style = {at={(ticklabel* cs:1.03)},anchor=west},
    xtick = {0,0.25,0.5,0.75,1},
    ]
    \addplot[color=red,mark=none,dashed,semithick] table[x=e,y=u]{mixed.dat};
    \addplot[color=blue,mark=none,semithick] table[x=e,y=l]{mixed.dat};
    \addplot[color=darkgreen,mark=none,dashdotted,semithick] table[x=e,y=o]{mixed.dat};
    \legend{$-\sigma_1\invP{\eps}$,$L$,$-\sigma_2\invP{\eps}$};
    \end{axis}
\end{tikzpicture}
\caption{Plot of the second order asymptotic rate $L$ (blue-solid) defined in \eqref{eq:L-relation} and bounds on $L$ (red-dashed and green-dash-dotted) for $\eps\in(0,1/2)$ \eqref{eq:L-bounds-left} and $\eps\in(1/2,1)$ \eqref{eq:L-bounds-right} for a mixed source $(\rho_1,\rho_2,t)$ with the values $\sigma_1=0.235$, $\sigma_2=0.712$, and $t=0.425$.}
\label{fig:L-plot}
\end{figure}
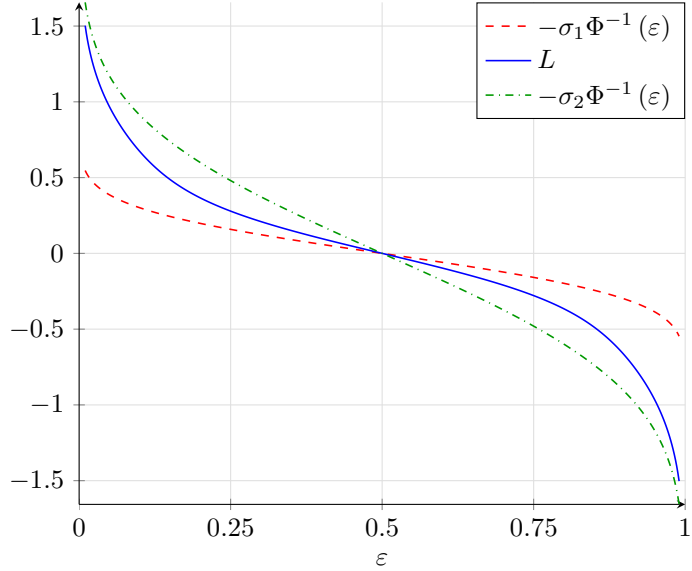

\section{Proofs}\label{sec:proofs}
The following lemma is a direct consequence of \Cref{thm:SOA} and a key ingredient in the proof of \Cref{thm:general-mixture}.
\begin{lemma}\label{lem:asym-norm-dominance}
Let $\rho_1,\rho_2\in\cD(\cH)$ with $S_i\coloneqq S(\rho_i)$ and $\sigma_i\coloneqq \sigma(\rho_i)$ for $i=1,2$. If $S_1 > S_2$, then for any constant $C$ we have:
\begin{align}
\lim_\nlim \tr\left(\rho_2^{\otimes n} \left\lbrace \rho_2^{\otimes n}\leq 2^{-nS_1-\sqrt{n} C}\one \right\rbrace\right) &= 0\label{eq:asym-norm-dominance-0}\\
\lim_\nlim \tr\left(\rho_1^{\otimes n} \left\lbrace \rho_1^{\otimes n}\leq 2^{-nS_2-\sqrt{n} C}\one \right\rbrace\right) &= 1.\label{eq:asym-norm-dominance-1}
\end{align}
\end{lemma}
\begin{proof}
In order to prove \eqref{eq:asym-norm-dominance-0}, define $f_n\coloneqq \sqrt{n}(S_1-S_2)$ and note that $f_n \xrightarrow{\nlim} \infty$ by assumption. We then obtain the following bound for some constant $K>0$:
\begin{align*}
\tr\left(\rho_2^{\otimes n} \left\lbrace \rho_2^{\otimes n}\leq 2^{-nS_1-\sqrt{n} C}\one \right\rbrace\right) &= \tr\left(\rho_2^{\otimes n} \left\lbrace \rho_2^{\otimes n}\leq 2^{-nS_2-\sqrt{n}(C + f_n)}\one \right\rbrace\right)\\
&\leq \Phi\left(-\frac{C+f_n}{\sigma_2}\right) + \frac{K}{\sqrt{n}}
\end{align*}
where the inequality follows from \Cref{thm:SOA}. This yields \eqref{eq:asym-norm-dominance-0} since $\lim_{x\rightarrow -\infty}\Phi(x)=0$. Identity \eqref{eq:asym-norm-dominance-1} is proved along similar lines.
\end{proof}

We also state the following result by Hayashi \cite{Hay02}, which gives an upper bound on the ensemble average fidelity. 
For a proof in our notation, see Proposition 7 in Section V.A of \cite{DL14b}.
\begin{lemma}[\cite{Hay02}]\label{prop:averaged-fidelity-bound}
Let $\kE=\lbrace p_i,\psi_i\rbrace_i$ be an ensemble of pure states and set $\rho=\sum_ip_i\psi_i$. Let $\cV\colon \lbrace i\rbrace\rightarrow \cD(\cH_c)$ be a visible encoding map with $\cH_c$ denoting the compressed Hilbert space with $\dim\cH_c = M$, and 
let $\kD\colon \cD(\cH_c)\rightarrow \cD(\cH)$ denote the decoding CPTP map. Then for the code $\cC=(\cV,\kD,M)$ we have
\begin{align*}
\bar{F}(\kE,\cC)\leq \max\lbrace \tr (P\rho)\colon P\text{ is a projection on }\cH\text{ with }\tr P=M\rbrace.
\end{align*}
\end{lemma}

We can now prove an upper bound on the ensemble average fidelity that we need for proving the converse bounds of \Cref{thm:general-mixture} and \Cref{thm:refined-analysis}.
\begin{lemma}\label{lem:fidelity-converse-bound}
Let $(\lbrace\rho_j\rbrace_{j=1}^k,\lbrace t_j\rbrace_{j=1}^k)$ be a mixed source with corresponding source state $\rho=\sum_{j=1}^{k} t_j\rho_j$ and ensemble $\kEm$ defined in \eqref{eq:mixed-source-ensemble} for $n=1$. For any code $\cC=(\cV,\kD,M)$ and $\gamma\in\mathbb{R}$, the ensemble average fidelity satisfies 
\begin{align*}
\bF(\kEm,\cC) \leq 1 - \sum_{j=1}^k t_j \tr(\rho_j\lbrace\rho_j\leq 2^{-\gamma}\one\rbrace) + 2^{-\gamma+\log M}.
\end{align*}
\end{lemma}
\begin{proof}
By \Cref{prop:averaged-fidelity-bound} there is a projection $Q$ with $\tr Q = M$ such that $\bF(\kEm,\cC)\leq \tr (Q\rho)$. For arbitrary $\gamma\in\mathbb{R}$, we then compute:
\begin{align*}
\bF(\kEm,\cC) &\leq \tr (Q\rho)\\
&= \sum_{j=1}^k t_j \tr Q\rho_j\\
&= \sum_{j=1}^k t_j \tr [Q(\rho_j-2^{-\gamma}\one)] + 2^{-\gamma}\tr Q\\
&\leq \sum_{j=1}^k t_j \tr[\lbrace \rho_j> 2^{-\gamma}\one_n\rbrace (\rho_j-2^{-\gamma}\one)] + 2^{-\gamma + \log M}\\
&= 1 - 2^{-\gamma}\tr\one - \sum_{j=1}^k t_j \tr(\lbrace\rho_j\leq 2^{-\gamma}\one\rbrace (\rho_j-2^{-\gamma}\one)) + 2^{-\gamma + \log M}\\
&= 1 - 2^{-\gamma}\tr\one - \sum_{j=1}^k t_j \tr(\rho_j\lbrace\rho_j\leq 2^{-\gamma}\one\rbrace)\\
&\qquad {} + 2^{-\gamma}\sum_{j=1}^k t_j  \tr\lbrace\rho_j\leq 2^{-\gamma}\one\rbrace + 2^{-\gamma + \log M}\\
&\leq 1 - \sum_{j=1}^k t_j \tr(\rho_j\lbrace\rho_j\leq 2^{-\gamma}\one\rbrace) + 2^{-\gamma+\log M}
\end{align*}
where we used \Cref{lem:tr-projector} in the second inequality, the identity $\lbrace \rho_j> 2^{-\gamma}\one\rbrace = \one - \lbrace\rho_j\leq 2^{-\gamma}\one\rbrace$ in the third equality, and $\lbrace\rho_j \leq 2^{-\gamma}\one\rbrace \leq \one$ in the last inequality.
\end{proof}

We also record the following simple observation: Let $A,B,C\in\cP(\cH)$ be pairwise commuting operators with $B\leq C$. Then we have
$\lbrace A\leq B\rbrace \leq \lbrace A\leq C\rbrace$,
which can easily be seen to be true by considering a common eigenbasis of $A$, $B$, and $C$ and checking the corresponding relation in the scalar case. We will use this result in the following form:
\begin{lemma}\label{lem:brace-projectors}
Let $a,b\in\mathbb{R}$ with $a\leq b$, then for any $X\geq 0$ we have 
\begin{align*}
\lbrace X\leq 2^{-b}\one \rbrace \leq \lbrace X\leq 2^{-a}\one\rbrace.
\end{align*}
\end{lemma}

For the remainder of this section, we abbreviate $\rho^n\equiv \rho^{\otimes n}$.

\subsection{Universal source code achieving second order asymptotic rates}\label{sec:universal}
In this section we construct a universal source code that, given parameters $a\in\mathbb{R}$ (which is to be chosen later as the first order rate) and $\eps\in(0,1)$, achieves a second order asymptotic rate $b(a,\eps|\rho)$ for any $\rho\in\cD(\cH)$. Our construction relies on ideas taken from papers by Jozsa \emph{et al.~}\cite{JHHH98} and Hayashi \cite{Hay08}.

Let $\cX=\lbrace 1,\dots,d\rbrace$. The type $P_{\underline{x}}$ of a sequence $\underline{x}=x_1\dots x_n\in\cX^n$ is the empirical distribution of the letters of $\cX$ in $\underline{x}$, that is, $P_{\underline{x}}(x) = \frac{1}{n}\sum_{i=1}^n \delta_{x_i,x}$ for all $x\in\cX$. We denote by $\cT_n$ the set of all types, and for a type $P\in\cT_n$ we denote by $T_P^n\subset \cX^n$ the set of sequences of type $P$. Following \cite{Hay08}, for $a,b\in\mathbb{R}$ we define
\begin{align*}
T_n(a,b) \coloneqq \bigcup\left\lbrace T_P^n\colon P\in\cT_n\text{ with }|T_P^n|\leq 2^{an+b\sqrt{n}} \right\rbrace  \subset \cX^n.
\end{align*}
A simple type-counting argument \cite{CK11} shows that 
\begin{align*}
\left| T_n(a,b)\right| \leq (n+1)^d 2^{an+b\sqrt{n}}.
\end{align*}
Let now $B=\lbrace |e_1\rangle,\dots,|e_d\rangle\rbrace$ be a basis of $\cH$. As in \cite{JHHH98}, we define the subspace 
\begin{align*}
\Xi^n_{a,b}(B) \coloneqq \spn\lbrace |e_{\ul{i}}\rangle\in B^{\ox n}\colon \ul{i}\in T_n(a,b)\rbrace,
\end{align*}
that is, $\Xi^n_{a,b}(B)$ is the span of basis vectors of the product basis $B^{\ox n}$ of $\cH^{\ox n}$ labelled by sequences in $T_n(a,b)$. The code space $\Upsilon^n_{a,b}$ of the universal source code is now obtained by varying $B$ over all bases of $\cH$. More precisely, we define $\Upsilon^n_{a,b}$ as the smallest subspace of $\cH^{\ox n}$ containing $\Xi^n_{a,b}(B)$ for all bases $B$ of $\cH$. To estimate the size of $\Upsilon^n_{a,b}$, we use the following 

\begin{lemma}[\cite{JHHH98}]\label{lem:dim-H-phi}
Let $|\phi\rangle\in\cH^{\ox n}$ with $\dim\cH=d$, and let $\cH_\phi\coloneqq \spn\lbrace A^{\ox n}|\phi\rangle \colon A\in \cB(\cH)\rbrace$, then $\dim \cH_\phi \leq (n+1)^{d^2}$.
\end{lemma}

We now obtain:

\begin{lemma}\label{lem:size-of-code-space}
With the above definitions, the dimension of the code space $\Upsilon^n_{a,b}\subseteq\cH^{\ox n}$ can be estimated as
\begin{align*}
\dim \Upsilon^n_{a,b} \leq (n+1)^{d^2+d} 2^{an+b\sqrt{n}}.
\end{align*}
\end{lemma}

\begin{proof}
Here, we closely follow an argument in \cite{JHHH98}. First, let $B_0$ be a fixed basis of $\cH$. Then any other basis $B$ can be obtained from $B_0$ by applying some unitary operator $U$ on the basis vectors of $B_0$. As $\Xi^n_{a,b}(B)$ is the span of tensor products of elements in $B$, we have
\begin{align*}
\Xi^n_{a,b}(B) = \lbrace U^{\ox n} |\phi\rangle\colon |\phi\rangle\in \Xi^n_{a,b}(B_0)\rbrace.
\end{align*}
Hence, the following holds for the code space $\Upsilon^n_{a,b}$:
\begin{align*}
\Upsilon^n_{a,b} &= \spn\lbrace U^{\ox n}|\phi\rangle\colon U\in\cU(d), |\phi\rangle\in \Xi^n_{a,b}(B_0)\rbrace\\
&\subset \spn\lbrace A^{\ox n}|\phi\rangle\colon A\in \cB(\cH), |\phi\rangle\in \Xi^n_{a,b}(B_0)\rbrace
\end{align*}
As $\dim\Xi^n_{a,b}(B_0) \leq |T_n(a,b)|\leq (n+1)^d 2^{an+b\sqrt{n}}$, the claim now follows from \Cref{lem:dim-H-phi}.
\end{proof}

\begin{proposition}[Universal code achieving second order rate]\label{prop:universal-code}
Let $\kE=\lbrace p_i,\psi_i\rbrace_i$ be the pure-state ensemble of an arbitrary memoryless quantum source with associated source state $\rho\in\cD(\cH)$, and abbreviate $S\equiv S(\rho)$ and $\sigma\equiv \sigma(\rho)$.
Let $\Pi_n$ be the projector onto the code space $\Upsilon^n_{S,b}$ defined as above, and consider the visible encoding map
\begin{align}
\cV_n\colon \ul{i} &\longmapsto \frac{\Pi_n \psi_{\ul{i}}\Pi_n}{\tr(\Pi_n\psi_{\ul{i}})}.\label{eq:visible-encoder}
%\cE_n\colon \psi_{\ul{i}} &\longmapsto \Pi_n \psi_{\ul{i}}\Pi_n + \tr(\psi_{\ul{i}}(I_n-\Pi_n)) \varphi \label{eq:blind-encoder}
\end{align}
%where $\ul{i}\coloneqq i_1\dots i_n$ and $\psi_{\ul{i}}\coloneqq \psi_{i_1}\ox\dots\ox \psi_{i_n}$ is a pure state in the source ensemble.
%In \eqref{eq:blind-encoder} $\varphi$ is an arbitrary pure state in $\im\Pi_n$. 
We set $M_n\coloneqq \dim \Upsilon^n_{S,b}$, and define the decoding operation $\kD_n\colon \Upsilon^n_{S,b}\rightarrow \cH^{\ox n}$ as the trivial embedding. For $n$ uses of the source $\rho$, the sequence $\lbrace \cC_n\rbrace_{n\in\mathbb{N}}$ of codes $\cC_n=(\cV_n,\kD_n,M_n)$ then achieves the second order rate $b=b(S,\eps|\rho)$, where $\eps = 1-\Phi(b/\sigma)$.
\end{proposition}

\begin{proof}
\Cref{lem:size-of-code-space} immediately yields
\begin{align*}
\limsup_{n\rightarrow\infty} \frac{\log M_n-Sn}{\sqrt{n}} \leq \limsup_{n\rightarrow\infty} \frac{(d^2+d)\log (n+1)}{\sqrt{n}} + b = b.
\end{align*}
With the visible encoding given by \eqref{eq:visible-encoder}, we can express the ensemble average fidelity $\bF(\kE^n,\cC_n)$ as \cite[Sect.~V.A.3]{DL14b}
\begin{align}\label{eq:avg-fidelity-expression}
\bF(\kE^n,\cC_n) &= \tr(\rho^n\Pi_n).
\end{align}
We now employ the following relation proved by Hayashi \cite{Hay08} in the context of classical fixed-length source coding:
\begin{align*}
\cS_n\coloneqq \left\lbrace \underline{x}\in\cX^n\colon -\log P^n(\underline{x}) < na + \sqrt{n}b \right\rbrace \subseteq T_n(a,b)
\end{align*}
which holds for arbitrary $a,b\in\mathbb{R}$ and probability distributions $P$ with support on $\lbrace 1, \dots, d\rbrace$. 
Consider the spectral decomposition $\rho = \sum_i r_i |\varphi_i\rangle\langle\varphi_i|$, and set $P_\rho = \lbrace r_i\rbrace_i$ and $B_\rho = \lbrace |\varphi_i\rangle\rbrace_i$. 
Observe that the projector $\lbrace \rho^n > 2^{-na-\sqrt{n}b}\one_n\rbrace$ projects onto eigenvectors of $\rho^n$ labelled by elements of $\cS_n$, upon choosing $P=P_\rho$. 
Since the code space $\Upsilon^n_{a,b}$ includes the subspace $\Xi^n_{a,b}(B_\rho)$, we have the operator inequality
\begin{align}\label{eq:operator-inequality}
\Pi_n \geq \left\lbrace \rho^n > 2^{-na-\sqrt{n}b}\one_n \right\rbrace.
\end{align}
We now set $a=S$ in \eqref{eq:operator-inequality} and substitute it in \eqref{eq:avg-fidelity-expression}. Taking the limit inferior, we obtain
\begin{align*}
\liminf_{n\rightarrow\infty}\bF(\kE^n,\cC_n) &= \liminf_{n\rightarrow\infty}\tr(\rho^n\Pi_n)\\
&\geq \liminf_{n\rightarrow\infty}\tr\left(\rho^n \left\lbrace \rho^n > 2^{-nS-\sqrt{n}b}\one_n \right\rbrace\right)\\
&= 1 - \limsup_{n\rightarrow\infty}\tr\left(\rho^n \left\lbrace \rho^n \leq 2^{-nS-\sqrt{n}b}\one_n\right\rbrace \right)\\
&= 1 - \Phi\left(-\frac{b}{\sigma}\right)\\
&= \Phi\left(\frac{b}{\sigma}\right),
\end{align*}
where we used \Cref{thm:SOA} in the third equality. Setting $\eps\coloneqq 1-\Phi(b/\sigma)$ now yields the claim.
\end{proof}

\subsection{General mixture}
In this section we prove the assertion of \Cref{thm:general-mixture}, which states that for $a>0$ and $\eps\in(0,1)$ the second order asymptotic rate $b(a,\eps|\rho)$ for $n$ uses of a general mixed source $(\rho_{\lambda},\dl)_{\lambda\in\Lambda}$ with source state $\rho^{(n)} = \int_{\Lambda} \rho_{\lambda}^{\ox n}\dl$ is given by the solution of the relation
\begin{align}\label{eq:b-solution}
\int_{\cL_=(a)}\Phi\left(\frac{b}{\sigma_{\lambda}}\right) \dl + \int_{\cL_<(a)}\dl = 1-\eps.
\end{align}
Here, the sets $\cL_=(a)$ and $\cL_<(a)$ are defined by
\begin{align*}
\cL_=(a) &\coloneqq \lbrace \lambda\in \Lambda\colon S(\rho_\lambda) = a\rbrace & \cL_<(a) &\coloneqq \lbrace \lambda\in \Lambda\colon S(\rho_\lambda) < a\rbrace,
\end{align*}
and we set $\sigma_\lambda\coloneqq \sigma(\rho_{\lambda})$ (cf.~\Cref{def:quantum-relative-entropy}(ii)).
Before we proceed with the proof, we note that the converse bound on the ensemble average fidelity in \Cref{lem:fidelity-converse-bound} holds for arbitrary ensembles $\lbrace d\mu(\lambda),\psi_\lambda\rbrace_{\lambda\in\Lambda}$ with respect to the measure $\mu$ on $\Lambda$. 
Here, $\psi_\lambda\in\cD(\cH)$ is a pure state for $\lambda\in\Lambda$, and $\rho = \int_{\Lambda} \psi_\lambda d\mu(\lambda)$ is the corresponding ensemble average state.

\subsubsection{Converse bound}\label{sec:general-converse}
Denoting the solution of \eqref{eq:b-solution} by $b^*$, we first prove the converse statement, i.e.~$b(a,\eps|\rho) \geq b^*$.
To this end, assume that $R<b^*$ is an $(a,\eps)$-achievable second order rate, that is, there is a sequence $\lbrace \cC_n\rbrace_{\nin} $ of codes $\cC_n = (\cV_n,\kD_n,M_n)$ for $n$ uses of the mixed source $(\rho_{\lambda},\dl)_{\lambda\in\Lambda}$ with source ensemble $\kEm^{(n)}$ such that
\begin{subequations}
\begin{align}
\liminf_{n\rightarrow\infty}\bar{F}\left(\kEm^{(n)},\cC_n\right) &\geq 1-\eps \label{eq:general-ach-condition1}\\
\limsup_{n\rightarrow\infty}\frac{\log M_n - na}{\sqrt{n}} &\leq R.\label{eq:general-ach-condition2}
\end{align}
\end{subequations}
Choose $\delta>0$ such that $R+2\delta < b^*$.
Then by \eqref{eq:general-ach-condition2} we have for sufficiently large $n$ that
\begin{align}\label{eq:general-M_n}
\log M_n < na + \sqrt{n}(R+\delta).
\end{align}
\Cref{lem:fidelity-converse-bound} yields the following bound on the fidelity $\bar{F}\left(\kEm^{(n)},\cC_n\right)$ for arbitrary $\gamma\in\mathbb{R}$:
\begin{align*}
\bar{F}\left(\kEm^{(n)},\cC_n\right) \leq 1 - \int_{\Lambda}\tr\left( \rho_{\lambda}^n\lbrace \rho_{\lambda}^n\leq 2^{-\gamma}\one_n\rbrace\right) \dl + 2^{-\gamma+\log M_n}.
\end{align*}
We now set $\gamma = \log M_n + \sqrt{n}\delta$, such that by \eqref{eq:general-M_n} we have
\begin{align*}
\gamma < na + \sqrt{n}(R+2\delta).
\end{align*} 
Hence, \Cref{lem:brace-projectors} yields
\begin{align*}
\bar{F}\left(\kEm^{(n)},\cC_n\right) &\leq 1 - \int_{\Lambda}\tr\left( \rho_{\lambda}^n \left\lbrace \rho_{\lambda}^n\leq 2^{-na-\sqrt{n}(R+2\delta)}\one_n \right\rbrace \right) \dl + 2^{-\sqrt{n}\delta}\\
&= 1 + 2^{-\sqrt{n}\delta} - \int_{\cL_=(a)}\tr\left( \rho_{\lambda}^n \left\lbrace \rho_{\lambda}^n\leq 2^{-na-\sqrt{n}(R+2\delta)}\one_n \right\rbrace \right) \dl\\
& \qquad\qquad {} - \int_{\cL_<(a)}\tr\left( \rho_{\lambda}^n \left\lbrace \rho_{\lambda}^n\leq 2^{-na-\sqrt{n}(R+2\delta)}\one_n \right\rbrace \right) \dl\\
&\qquad \qquad {} - \int_{\cL_>(a)}\tr\left( \rho_{\lambda}^n \left\lbrace \rho_{\lambda}^n\leq 2^{-na-\sqrt{n}(R+2\delta)}\one_n \right\rbrace \right) \dl\numberthis\label{eq:fidelity-bound}
\end{align*}
where we defined $\cL_>(a) \coloneqq \lbrace\lambda\in\Lambda\colon S(\rho_\lambda) > a\rbrace$.
By \Cref{thm:SOA} and \Cref{lem:asym-norm-dominance} we have the following:
\begin{align}\label{eq:liminf-cases}
\lim_{n\to\infty} \tr\left( \rho_{\lambda}^n \left\lbrace \rho_{\lambda}^n\leq 2^{-na-\sqrt{n}(R+2\delta)}\one_n \right\rbrace \right) = \begin{cases}
\Phi\left( \frac{-(R+2\delta)}{\sigma_\lambda} \right) & \text{if } S(\rho_\lambda)=a\\
1 & \text{if } S(\rho_{\lambda}) > a\\
0 & \text{if } S(\rho_{\lambda}) < a
\end{cases} 
\end{align}
Taking the limit inferior on both sides of \eqref{eq:fidelity-bound}, noting that we can exchange limit and integral by the Dominated Convergence Theorem, and using \eqref{eq:liminf-cases}, we obtain
\begin{align*}
\liminf_{n\to\infty}\bar{F}\left(\kEm^{(n)},\cC_n\right) &\leq 1 - \int_{\cL_=(a)} \Phi\left(\frac{-(R+2\delta)}{\sigma_\lambda}\right) \dl - \int_{\cL_>(a)} \dl \\
&= 1 - \int_{\cL_=(a)}\dl + \int_{\cL_=(a)}\Phi\left(\frac{R+2\delta}{\sigma_\lambda}\right) \dl - \int_{\cL_>(a)} \dl\\
&= \int_{\cL_<(a)}\dl + \int_{\cL_=(a)}\Phi\left(\frac{R+2\delta}{\sigma_\lambda}\right) \dl\\
&< \int_{\cL_<(a)}\dl + \int_{\cL_=(a)}\Phi\left(\frac{b^*}{\sigma_\lambda}\right) \dl\\
&= 1-\eps.
\end{align*}
Here, we used the relation $\Phi(-x)=1-\Phi(x)$ in the first equality, the fact that $\mu$ is a normalized measure on $\Lambda = \cL_=(a)\cup\cL_<(a)\cup \cL_>(a)$ in the second equality, and the assumption $R+2\delta < b^*$ in the strict inequality.
This is a contradiction to \eqref{eq:general-ach-condition1}, and hence, we have $b(a,\eps|\rho)\geq b^*$.

\subsubsection{Achievability bound}
We now use the universal source code $\lbrace\cC_n\rbrace_{n\in\mathbb{N}}$ with $\cC_n\coloneqq \lbrace\cV_n,\kD_n,M_n\rbrace$ as defined in \Cref{prop:universal-code} to prove that the second order rate $b^*$ is achievable.
To this end, consider $n$ uses of a mixed source $(\rho_\lambda,d\mu(\lambda))_{\lambda\in\Lambda}$ with source state $\rho^{(n)}$ as defined in \eqref{eq:general-mixture} and ensemble $\kEm^{(n)}$ as defined in \eqref{eq:general-mixture-ensemble}.
Recall that $\Pi_n$ denotes the projector onto the code space $\Upsilon^n_{a,b}$ defined in \Cref{sec:universal}. 
For arbitrary $a>0$, the calculation from \cite[Sect.~V.A.3]{DL14b} shows that we can express the ensemble average fidelity $\bF(\kEm^{(n)},\cC_n)$ as
\begin{align*}
\bF\left(\kEm^{(n)},\cC_n\right) &= \tr\left(\Pi_n\rho^{(n)}\right)\\
&= \int_{\Lambda} \tr\left(\Pi_n \rho_\lambda^n\right) d\mu(\lambda)\\
&\geq \int_{\Lambda} \tr\left(\rho_\lambda^n \left\lbrace \rho_\lambda^n > 2^{-na-\sqrt{n}b}\one_n \right\rbrace \right)d\mu(\lambda)\\
&= 1 - \int_{\Lambda} \tr\left(\rho_\lambda^n \left\lbrace \rho_\lambda^n \leq 2^{-na-\sqrt{n}b}\one_n \right\rbrace \right) d\mu(\lambda),\numberthis\label{eq:ach-starting-point}
\end{align*}
where the inequality follows from \eqref{eq:operator-inequality}.
We set $b=b^*$, where $b^*$ is once again defined as the solution of the relation
\begin{align*}
\int_{\cL_=(a)}\Phi\left(\frac{b}{\sigma_{\lambda}}\right) \dl + \int_{\cL_<(a)}\dl = 1-\eps.
\end{align*}
Similar to \Cref{sec:general-converse}, we then compute
\begin{align*}
&\liminf_{n\to\infty} \int_{\Lambda} \tr\left(\rho_\lambda^n \left\lbrace \rho_\lambda^n \leq 2^{-na-\sqrt{n}b^*}\one_n\right\rbrace \right) d\mu(\lambda)\\
&\qquad\qquad {} = \liminf_{n\to\infty} \int_{\cL_=(a)} \tr\left(\rho_\lambda^n \left\lbrace \rho_\lambda^n \leq 2^{-na-\sqrt{n}b^*}\one_n\right\rbrace \right) d\mu(\lambda)\\
&\qquad\qquad\qquad {} + \liminf_{n\to\infty} \int_{\cL_<(a)} \tr\left(\rho_\lambda^n \left\lbrace \rho_\lambda^n \leq 2^{-na-\sqrt{n}b^*}\one_n\right\rbrace \right) d\mu(\lambda)\\
&\qquad\qquad\qquad {} + \liminf_{n\to\infty} \int_{\cL_>(a)} \tr\left(\rho_\lambda^n \left\lbrace \rho_\lambda^n \leq 2^{-na-\sqrt{n}b^*}\one_n\right\rbrace \right) d\mu(\lambda)\\
&\qquad\qquad {} = \int_{\cL_=(a)} \Phi\left(\frac{b^*}{\sigma_\lambda}\right) d\mu(\lambda) + \int_{\cL_<(a)}d\mu(\lambda)\\
&\qquad\qquad {} = 1-\eps,
\end{align*}
where the exchange of the limit inferior and the integral is permitted by the Dominated Convergence Theorem, and we once again used \eqref{eq:liminf-cases}.
Hence, we obtain $\liminf_{n\to\infty}\bF(\kEm^{(n)},\cC_n) \geq 1-\eps$ by \eqref{eq:ach-starting-point}.
Moreover, \Cref{lem:size-of-code-space} yields that the universal source code $\{\cC_n\}_{n\in\mathbb{N}}$ satisfies
\begin{align*}
\limsup_{n\to\infty}\frac{\log M_n-na}{\sqrt{n}} \leq b^*.
\end{align*}
Hence, the rate $b^*$ is achievable, and we obtain $b(a,\eps|\rho)\leq b^*$. 
Together with $b(a,\eps|\rho)\geq b^*$ from the preceding section, this proves \Cref{thm:general-mixture}.

\subsection{Mixed source consisting of two memoryless sources}
In this section, we prove the second order asymptotic rates for $n$ uses of a mixed source $(\rho_1,\rho_2,t)$ consisting of two memoryless sources $\rho_1$ and $\rho_2$, as stated in \Cref{thm:refined-analysis}.
We set $S_i=S(\rho_i)$ and $\sigma_i= \sigma(\rho_i)$ for $i=1,2$.
By \Cref{thm:identical-entropies}, we have the relation
\begin{align}
\sum_{i\colon S_i=a} t_i \Phi\left(\frac{L}{\sigma_i}\right) + \sum_{i\colon S_i<a} t_i = 1-\eps. \label{eq:finite-k-relation}
\end{align}
In the first case of \Cref{thm:refined-analysis}, where $S_1 = S_2 = S$, we set $a=S$ in \eqref{eq:finite-k-relation}, which immediately yields 
\begin{align*}
t \Phi\left(\frac{L}{\sigma_1}\right) + (1-t) \Phi\left(\frac{L}{\sigma_2}\right) = 1-\eps,
\end{align*}
and thus proves \Cref{thm:refined-analysis}(i).

Consider now the second case of \Cref{thm:refined-analysis}, where $S_1>S_2$ and $t>\eps$. 
Choosing $a=S_1$, we obtain from \eqref{eq:finite-k-relation} that
\begin{align*}
t \Phi\left(\frac{b^*}{\sigma_1}\right) + 1-t &= 1-\eps,
\end{align*}
which implies that
\begin{align*}
b^* &= \sigma_1 \invP{1-\frac{\eps}{t}} = -\sigma_1\invP{\frac{\eps}{t}}.
\end{align*}
This is the assertion of \Cref{thm:refined-analysis}(ii).

Finally, we consider the third case of \Cref{thm:refined-analysis}, where $S_1>S_2$ and $t<\eps$.
Choosing $a=S_2$ in \eqref{eq:finite-k-relation} yields
\begin{align*}
b^* = \sigma_2 \invP{\frac{1-\eps}{1-t}} = \sigma_2 \invP{1-\frac{\eps-t}{1-t}} = -\sigma_2  \invP{\frac{\eps-t}{1-t}},
\end{align*}
and this proves \Cref{thm:refined-analysis}(iii).

\section{Conclusions and open questions}\label{sec:conclusion}
We derived the second order asymptotic rates of fixed-length visible quantum source coding using a mixed source consisting of memoryless sources. To our knowledge, this is the first example of a second order asymptotic analysis of the optimal rate for a {\em{quantum}} information-processing task which uses a resource with memory. Previously, such analyses in the quantum setting were restricted to memoryless (or i.i.d.) resources \cite{TH13,Li14,KH13,KH13a,DL14b,BG13}.

An interesting problem is to extend our methods to mixed classical-quantum channels. In the classical case this has been studied by Polyanskiy \textit{et al.}~\cite{PPV11} (see also \cite{TT14}). The main result about the second order expansion of the capacity of a mixed channel (\cite[Thm.~7]{PPV11}) bears a close resemblance to the equivalent result about source coding using a mixed source as in \cite{NH13}.

\section*{Acknowledgements}
\addcontentsline{toc}{section}{Acknowledgements}
We would like to thank Vincent Tan for useful discussions, and the anonymous referees for helpful feedback.

\printbibliography[title={References},heading=bibintoc]
\end{document}